\documentclass[11pt, letterpaper]{article}
\usepackage{times}
\usepackage{latexsym}
\usepackage{amsfonts,amsthm,amssymb}
\usepackage{amsmath}
\usepackage{euscript}
\usepackage{amstext}
\usepackage{graphicx}
\usepackage{color}
\usepackage{url,hyperref}
\usepackage{verbatim}
\usepackage{xspace}
\usepackage{framed}

\newtheorem{lemma}{Lemma}[section]
\newtheorem{theorem}[lemma]{Theorem}

\newtheorem{observation}{Observation}

        {\hspace*{\fill}$\Box$\par}

\newcommand{\size}{\textrm{Size}}

\newcommand{\initOneLiners}{%
    \setlength{\itemsep}{0pt}
    \setlength{\parsep }{pt}
    \setlength{\topsep }{0pt}
}

\newcommand{\lpp}{\mbox{$\mathsf{LP}_\mathsf{new}$}\xspace}

\newcommand{\lpl}{\mbox{$\mathsf{LP(\ell)}$}\xspace}

\newcommand{\el}{{\ell}}	

\newcommand{\vol}{\textsf{Vol}}
\begin{document}

\title{ Minimizing Flow-Time on Unrelated Machines}
\author{ Nikhil Bansal \thanks{Department of Mathematics and Computer Science, Eindhoven University of Technology, Netherlands. {\tt n.bansal@tue.nl}. Supported
by the NWO Vidi grant 639.022.211 and an ERC consolidator grant 617951.} \and Janardhan Kulkarni \thanks{Department of Computer Science, Duke University , 308 Research Drive, Durham, NC 27708. {\tt kulkarni@cs.duke.edu}. Supported by NSF Awards CCF-0745761, CCF-1008065, and CCF-1348696.}  }

\date{}
\maketitle

\thispagestyle{empty}

\begin{abstract}

We consider some classical flow-time minimization problems in the unrelated machines setting. In this setting, there is a set of $m$ machines and a set of $n$ jobs, and each job $j$ has a machine dependent processing time of $p_{ij}$ on machine $i$.  The flow-time of a job is the amount of time  the job spends in a system (its completion time minus its arrival time), and is one of the most  natural measures of quality of service. We show the following two results: an $O(\min(\log^2 n,\log n \log P))$ approximation algorithm for minimizing the total flow-time, and an $O(\log n)$ approximation for minimizing the maximum flow-time.
Here $P$ is the ratio of maximum to minimum job size. 
These are the first known poly-logarithmic guarantees for both the problems.
\end{abstract}

\newpage
\setcounter{page}{1}

\section{Introduction}
	\label{sec:intro}

\begingroup
\allowdisplaybreaks

Scheduling a set of jobs over a heterogeneous collection of machines to optimize some quality of service (QoS) measure is one of the central questions in approximation and scheduling theory. In modern computing environments be it web-servers, data-centers, clusters of machines or personal computers, heterogeneity of the processors and architectures is ubiquitous.
The most general and widely studied model that incorporates the heterogeneity of jobs and machines is the so-called {\em unrelated machines} setting. 
Here, there is a set $J$ of $n$ jobs and a set $M$ of $m$ machines. Each job $j$ is specified by its release time (or arrival time) $r_j$, which is the first time instant it is available for processing, and a machine-dependent processing requirement $p_{ij}$, which is the time taken to process $j$ on machine $i$. 

Besides the practical motivation, exploring basic scheduling problems in the unrelated machines setting has also led to 
the development of several fundamental techniques in algorithm design, for example \cite{LenstraST90, SkutellaJACM01, BansalSviridenko06, ChuzhoyKhanna, Svensson12, Garg09, AnandGK12}. 
However, one problem where no non-trivial approximation ratios are known
is that of minimizing the total flow-time on unrelated machines \cite{Naveen2006, amit, sitters}.
The flow-time of a job, defined as the amount of time the job spends in the system, is one of the most natural measures of quality of service, and is also sometimes referred to as response time or sojourn time. 
More precisely, if a job $j$ completes its processing at time $C_j$, then flow-time of the job $F_j$ is  defined $C_j -r_j$; i.e., its completion time minus arrival time.

We consider two natural flow-time based objectives in the unrelated machines setting: (i) Minimizing the total (or sum) flow-time of jobs and (ii) Minimizing the maximum flow-time.
Both these objectives have been studied quite extensively (as we discuss below). However, all the previously known results only hold in more restricted settings and the general case was 
open \cite{amit,sitters}. 

Our main results are the following.

\begin{theorem}
\label{thm:sumofflowtimes}
There exists a polynomial time $O(\log n \cdot \log P)$-approximation algorithm for minimizing the total flow time in the unrelated machine setting.
\end{theorem}

Here, $P$ denotes the ratio of maximum processing length of jobs to the minimum processing length. Using a standard trick  
this implies an $O(\log^2 n)$ approximation, which may be better if $P$ is super-polynomial in $n$. An approximation hardness of $\Omega(\log P)$ is also known for the problem even in the much simpler setting of identical machines \cite{gargIsaac}.

\begin{theorem}
\label{thm:maxflow}
There is an $O(\log n)$-approximation algorithm for minimizing the maximum flow-time in the unrelated machine setting.
In fact, in the schedule produced by our algorithm the maximum flow-time exceeds the optimum value by an additive $O(\log n) p_{\max}$ term, where $p_{\max}$ is the maximum size of a job in the optimum schedule.
\end{theorem}

Our algorithms are based on applying the iterated rounding technique and are quite different from the previous approaches to these problems. For the total flow-time objective, a key idea is to write a new time-indexed linear programming (LP) formulation for the  problem. 
The formulation we consider 
has much fewer constraints than the standard time-indexed LP formulation. Having fewer constraints is crucial in being able to use  iterated rounding technique.
We describe the new formulation and give an overview of the algorithm in section \ref{sec:unrelated}.
Theorem \ref{thm:sumofflowtimes} is proved in section \ref{sec:unrelated}. 

For the maximum flow-time problem, we follow a similar approach of solving a (different) LP relaxation with few constraints and then applying the iterated rounding technique.
However, there are some crucial differences between the two results, particularly in the rounding steps, as for maximum flow-time we must ensure that {\em no} job is delayed by too much.

\subsection{Related Work}

Scheduling to minimize flow-time has been extensively studied in the literature under various different models and objective functions and we only describe the work that is most relevant to our results. A more comprehensive survey of various flow-time related results can be found in \cite{Garg09,ImMP11,PST}.

\vspace{2mm}

{\bf Single Machine:}
Both total flow-time and maximum flow-time are well understood in the single machine case. The SRPT (Shortest Remaining Processing Time) algorithm is optimal for total flow-time if preemption is allowed, that is, when a job can be interrupted arbitrarily and resumed later from the point of interruption. Without preemptions, the problem becomes hard to approximate within $O(n^{1/2-o(1)})$ \cite{keller}. We will consider only preemptive algorithms in this paper.
For maximum flow-time,
First In First Out (FIFO) is easily seen to be an optimal (online) algorithm.

\vspace{2mm}

{\bf Multiple Machines:}
For multiple machines, various different settings have been studied.
The simplest is the identical machines setting, where a job has identical size on each machine ($p_{ij}=p_j$ for all $i$). A more general model is the related machines setting, where machine $i$ has speed $s_i$ and job $j$ has size  $p_j$ ($p_{ij} = p_j/s_i$).
Another model is the restricted assignment setting, where a job $j$ has a fixed size, but it can only be processed on a subset $S_j$ of machines ($p_{ij} \in \{p_j,\infty\}$). Clearly, all these are special cases of the unrelated machines setting.
As in most previous works, we will consider the non-migratory setting where a job must be executed on a single machine.

\vspace{2mm}

{\bf Results for Total flow-time:}
Leonardi and Raz \cite{LeonardiR07} obtained the first poly-logarithmic guarantee for identical machines and showed that SRPT is an $\Theta(\log(\min{(\frac{n}{m}, P)}))$ competitive algorithm.

Subsequently, other algorithms with similar competitive ratio, but other desirable properties such as no-migration and immediate dispatch were also obtained $\cite{AwerbuchALR02,Avrahami}$.
Later, poly-logarithmic offline and online guarantees were obtained for the related machines setting \cite{Naveen2006, amit}.
As mentioned previously, an $\Omega(\log^{1-\epsilon} P)$ hardness of approximation is known even for identical machines \cite{gargIsaac}.

The above approaches do not work for the restricted assignment case, which is much harder.
In an important breakthrough, Garg and Kumar \cite{GargK07} gave a $O(\log P)$ approximation, based on an elegant and non-trivial LP rounding approach. They consider a natural LP relaxation of the problem, and round it based on computing certain unsplittable flows \cite{dinitz} on an appropriately defined graph.

To extend these ideas to the unrelated machines case, \cite{gargIsaac} introduce a $(\alpha,\beta)$-variability setting (see \cite{gargIsaac} for details) and prove a general result that implies logarithmic approximations for both restricted assignment and related machines setting. 
For the unrelated setting, their result gives
an $O(k)$ approximation where $k$ is the number of different possible values of $p_{ij}$ in the instance.
Sitters \cite{sitters} also independently obtained a similar result using different techniques.
In general however these guarantees are polynomial in $n$ and $m$.

Interestingly, with job weights, approximating total weighted flow-time is $n^{\Omega(1)}$-hard even for identical machines \cite{ChekuriK02}. 
However, several interesting results are known for this measure in the resource augmentation setting \cite{ChekuriKhannaKumar04, Garg09, AnandGK12}. In this paper we only consider the unweighted setting.

\vspace{2mm}

{\bf Maximum flow-time:}
Relatively fewer results are known for the maximum flow-time problem.
For identical machines, Ambuhl and Mastrolilli \cite{maxflow} showed that FIFO (along with a simple greedy dispatch policy) is a $3$ competitive online algorithm.
More general settings were considered recently by Anand et al.~\cite{karl}, but all their positive results hold only in the resource augmentation setting.
For unrelated machines, Bansal gave \cite{Bansal05} a polynomial time approximation scheme for the case of $m=O(1)$. Prior to our work, no non-trivial approximation algorithm was known even for the substantially simpler related machine setting.

Maximum flow-time is also closely related to deadline scheduling problems: maximum flow-time is $D$ if and only if each job $j$ is completed by time $r_j + D$. However, usually the focus in deadline scheduling problems is to maximize the throughput (the jobs completed by their deadlines), and hence the results there do not
translate to our setting.

Also, observe that maximum flow-time is a (substantially harder) generalization of  minimizing makespan objective. In particular, maximum flow-time is equal to the makespan if all the jobs are released at the same time.

\subsection{Overview of Techniques}

Unlike other commonly studied measures such as makespan and completion time, 
a difficulty in approximating flow-time
is that it can be very sensitive to small changes in the schedule, and small errors can add up over time.
The following example is instructive.

Consider some hard instance $I$ of the makespan problem on unrelated machines with $n$ jobs and $m$ machines, such that in the optimum schedule all machines have load exactly $T$ (add dummy jobs if necessary). As the problem is strongly-NP hard, $T=\textrm{poly}(n,m)$.
On the other hand, in any schedule computed by an efficient algorithm, at least one unit of work will be left unfinished at time $T$.

Make $N$ copies of $I$, and create an instance of the maximum flow-time problem by releasing the jobs in the $i$-th copy at time $(i-1)T$, for $i=1,2,\ldots,N$.
Clearly, the optimum maximum flow-time is $T$ as the optimal schedule can finish jobs in the $i$-th copy before the next copy arrives.
On the other hand, any polynomial time algorithm must ensure that a backlog of work does not build up over the copies. Otherwise, the accumulated error at the end will be $\Omega(N)$, leading to a maximum-flow time of $\Omega(N/mT)$, which can be made arbitrarily large (say $N=n^2m^2T^2 = \textrm{poly}(n,m)$).

To get around such issues we adopt a two-step approach. First, we determine a coarse schedule by computing an assignment of each job to some machine and a time slot, and ensure that for any machine, no overload is created in any reasonably large time interval. 
This is done by formulating a suitable LP 
and then applying iterated rounding. A key property 
that enables us to apply iterated rounding is that our LP has few constraints. In the second step, we determine an actual schedule on each machine by scheduling the jobs according to SRPT or FIFO (depending on the problem), and show that the no-overload property ensures that the quality of the solution does not worsen substantially.
We believe that this approach is quite modular and should be useful for many other scheduling problems which involve writing time-indexed LP formulations.

\endgroup

\section{Minimizing the Total Flow-time}
\label{sec:unrelated}

\begingroup
\allowdisplaybreaks

In this section we prove Theorem \ref{thm:sumofflowtimes}. We start by describing our new LP relaxation for the problem.

\subsection{Alternate LP Relaxation and The High-level Idea}
\label{s:idea}

Before describing the new LP formulation that we use, we describe the standard time-indexed linear programming relaxation for the problem that was used, for example, in~\cite{GargK07,gargIsaac}.

\vspace{2mm}

\noindent{\bf Standard LP formulation:}
There is a variable $x_{ijt}$ for each machine $i \in [m]$, each job $j \in [n]$ and each unit time slot $t \geq r_{j}$. The $x_{ijt}$ variables indicate the amount to which a job $j$ is processed on machine $i$ during the time slot $t$.
The first set of  constraints (service constraints) \eqref{eqn:size} 
says that every job must be completely processed. The second set of constraints (capacity constraints)
\eqref{eqn:capacity1} 
enforces that a machine cannot process more than one unit of job during any time slot.
Note that this LP allows a job to be processed a job on multiple machines, and even at the same time.

\begin{align}
	\text{Minimize} \sum_{i,j,t} \left (\frac{t-r_j}{p_{ij}} + \frac{1}{2} \right) \cdot x_{ijt}   \nonumber \vspace{-3mm}
\end{align}
\vspace{-5mm}
\begin{align}
	 && \textrm{ s.t. }\hspace{2mm} \sum_{i} \sum_{t \geq r_j}\frac{ x_{ijt}}{p_{ij}} &\geq 1&& \forall j \label{eqn:size}  \\ 
	&& \sum_{j \, : \, t \geq r_{j}} x_{ijt} &\leq 1 &&\forall i,    t  \label{eqn:capacity1}  \\ 
	&& x_{ijt} &\geq 0 &&\forall i,j,t \geq 0  \nonumber
\end{align}

\vspace{1mm}
\noindent{\bf Fractional flow-time:}
The objective function needs explanation. The term $ \sum_{i,t} x_{ijt}$ is the total amount of processing done on job $j$.
The term $\sum_{i,t} (t-r_j) \cdot \frac{x_{ijt}}{p_{ij}}$ is the {\em fractional} flow-time of job $j$ and we  denote it by $f_j$. Recall that the (integral) flow-time of a job $j$ can be viewed as
summing up 1 over each time step that $j$ is alive, i.e.~$\sum_{t\geq r_j} \bf{1}_{(j \textrm{ is alive at } t)}$.
Similarly, the fractional flow-time is the sum over time of the remaining fraction of job $j$. On machine $i$, the fraction of job $j$ unfinished at time $t$ is $\sum_{t' > t} \frac{x_{ijt'}}{p_{ij}}$ (the numerator is the work done on $j$ on machine $i$ after $t$). Thus the fractional flow-time on machine $i$ is $\sum_{t \geq r_j} \sum_{t' > t} \frac{x_{ijt'}}{p_{ij}}$, which can be written as $\sum_t (t-r_j) \cdot \frac{x_{ijt}}{p_{ij}}$.
Note that the integral flow-time is at least the fractional flow-time plus half the size of a job, and thus the objective function in the LP above is valid lowerbound on flow-time. 
For more details on the LP above, see \cite{GargK07}.

 We assume that $\min_{i,j} p_{ij}\neq 0$ (otherwise $j$ can be scheduled on machine $i$ right upon arrival), 
and hence by scaling we assume henceforth that $\min_{i,j} p_{ij} =1$.  
We set $p_{ij} = \infty$ if an optimal solution to the time-indexed LP does not schedule the job $j$ on machine $i$ to any extent.  We say that a job $j$ can be assigned to machine $i$ if $p_{ij} \neq \infty$, and denote it by $j \rightarrow i$. 
 Define $P = \max_{i,j: j \rightarrow i } p_{ij}/\min_{i,j}p_{ij}$.
 For $k=1, 2,\ldots,\log P$, we say that a job $j$ belongs to class $k$ on machine $i$ if $p_{ij} \in (2^{k-1}, 2^{k}]$. Note that the class of a job depends on the machine.
We now describe the new LP relaxation for the problem.
The main idea is to ignore the capacity constraints \eqref{eqn:capacity1}
at each time slot, and instead only enforce them
over carefully chosen intervals of time.

Even though the number of constraints is fewer, as we will see, the quality of the relaxation is not sacrificed much.

\vspace{2mm}

\noindent{\bf New LP formulation:}
There is a variable $y_{ijt}$ (similar to $x_{ijt}$ before) that denotes the total units of job $j$ processed on machine $i$ at time $t$. (If a job $j$ has processing length $\infty$ on machine $i$, then $y_{ijt}$ variables are not defined.)
However, unlike the time-indexed relaxation, we allow $y_{ijt}$ to take values greater than one.
In fact, we will round the new LP in such a way that eventually $y_{ijt} = p_{ij}$ for each job, which will have a natural interpretation that job $j$ is scheduled at time $t$ on machine $i$.

For each class $k$ and each machine $i$, we partition the time horizon $[0, T]$ into intervals of size $4 \cdot 2^k$.
 Without loss of generality we can assume that $T \leq nP$; otherwise the input instance can be trivially split into two disjoint non-overlapping instances. For $a=1,2,\ldots$, let $I(i,a,k)$ denote the $a$-th interval of class $k$ on machine $i$. That is, $I(i,1,k)$ is the time interval $[0,4 \cdot 2^k]$ and $I(i,a,k) = ((4\cdot 2^k) (a-1),  (4 \cdot 2^k) a]$.
 We write the new LP relaxation.

\begin{align}
	\text{Minimize} \sum_{i} \sum_{t \geq r_{j}} \sum_{k} \sum_{j \in (2^{k-1}, 2^k]} \left(\frac{t- r_j}{p_{ij}} + \frac{1}{2} \right) \cdot y_{ijt} \tag{\lpp}  \label{Newprimal}
\end{align}
\vspace{-1mm}
\begin{align}
	&& \textrm{ s.t. }\quad \sum_{i} \sum_{t \geq r_j}\frac{ y_{ijt}}{p_{ij}} &\geq 1&& \forall j \label{eqn:service} \\
	&& \sum_{j \, : \, p_{ij} \leq 2^k} \sum_{t \in I(i,a,k)} y_{ijt} &\leq \size(I(i,a,k)) &&\forall i, k, a  \label{eqn:capacity} \\
	&& y_{ijt} &\geq 0  &&\forall i,  j,  t \, : \, t \geq r_{j} \nonumber
\end{align}

\vspace{2mm}

Here, $\size(I(i,a,k))$ denotes the size of the interval $I(i,a,k)$ which is $4 \cdot 2^k$ (but would change in later iterations of the LP when we apply iterated rounding).
Observe that in \eqref{eqn:capacity} only jobs of class $\leq k$ contribute to the left hand side of constraints corresponding to intervals of class $k$.

Clearly, \eqref{Newprimal} is a relaxation of the time indexed LP formulation considered above, as any valid solution there is also a valid solution to \eqref{Newprimal} (by setting $y_{ijt} = x_{ijt}$). Therefore, we conclude that an optimum solution to \eqref{Newprimal} lower bounds the value of an optimal solution.

\vspace{2mm}
\noindent \textbf{Remark:} When we do the iterative rounding and consider subsequent rounds, we will refer the intervals  $I(i,a,k)$  as $I(i,a,k,0)$.

\vspace{2mm}

\noindent{\bf The high-level approach:} The main idea of our algorithm is the following.
Let us call  a job $j$ to be {\em integrally assigned } to machine $i$ at time $t$, if $y_{ijt}=p_{ij}$ (note that this job will be completely executed on machine $i$). Let us view this as processing the job $j$ during $[t,t+p_{ij})$. In the algorithm,
we first find a {\em tentative} integral assignment of jobs to machines (at certain times)
such that the total flow-time of this solution is at most the LP value. This solution is tentative in the sense
that multiple jobs could use the same time slot; however we will ensure that the effect of this overlap is negligible.
More precisely, we show the following result.

\begin{lemma}
\label{lem:itr}
There exists a solution $y^* = \{y^*_{ijt}\}_{i,j,t}$ satisfying the following properties:
\begin{itemize}
\item ({\em Integrality:}) For each job $j$, there is exactly one  non-zero variable $y_{ijt}$ in $y^*$, which takes value $p_{ij}$. That is, each job is assigned integrally to exactly one machine, and one time slot : $y^*_{ijt} = p_{ij}$.
\item ({\em Low cost:}) The cost of $y^*$ is at most the cost of an optimal solution to \ref{Newprimal}.
\item ({\em Low overload:}) For any interval of time $[t_1, t_2]$, every machine $i$ and for every class $k$,
$$\displaystyle	\sum_{j \, : \, p_{ij} \leq 2^k} \sum_{t \in [t_1, t_2]} y^*_{ijt} \leq (t_2 - t_1) + O(\log n) \cdot 2^k.$$
That is, the total size of jobs of class at most $k$ assigned integrally in any time interval $[t_1,t_2]$ exceeds the size of the interval by at most $O(\log n) \cdot 2^k$.
\end{itemize}
\end{lemma}

 Lemma \ref{lem:itr} is the core of our algorithm, which will be proved using iterated rounding. In particular, we show, using a counting argument, that in each round a basic feasible optimum solution assigns at least a constant fraction of the jobs integrally.  Therefore, after $O(\log n)$ rounds every job is integrally assigned to some machine. In each round as some jobs get integrally assigned, we will fix them permanently and reduce the free space available in those intervals. Then, we merge these intervals greedily to ensure that the free space in an interval corresponding to class $k$ stays $O(1)\cdot2^k$. This merging process adds an overload of at most $O(1) \cdot 2^k$ to any time interval in each round. This ensures that the total error added for any time interval is $O(\log n) \cdot 2^k$.

The next step is to show that the tentative schedule can be converted to a valid preemptive schedule by increasing the total flow-time of the jobs by $O(\log P \log n)$ times the \lpp value. To this end, we use ideas similar to those used by \cite{amit,GargK07} for the related or restricted machines case. In particular, we schedule the jobs on each machine in the order given by the tentative schedule, while prioritizing the jobs in the shortest job first (SJF) order. The low overload property of the tentative schedule ensures that a job of class $k$ is additionally delayed by at most $O(\log n)\cdot 2^k$ due to jobs that arrive before it, or is delayed by smaller jobs (of strictly lower class) that arrive after the time when it is tentatively scheduled. In either case, we show that this delay can be charged to the total flow-time of other jobs.

\subsection{Tentative Schedule to Actual Schedule}
\label{s:tent}

We show how Theorem \ref{thm:sumofflowtimes} follows given a solution $y^*$ satisfying the conditions of Lemma \ref{lem:itr}.
Recall that in the solution $y^*$, for each job $j$, we have $y_{ijt} = p_{ij}$ for some time instant $t$ and some machine $i$, but this is not necessarily a valid schedule. We convert $y^*$ into a valid preemptive schedule $S$ as follows. Fix a machine $i$ and let $J(i,y^*)$
denote the set of jobs which are scheduled on machine $i$ in the solution $y^*$ (i.e.~jobs $j$ such that $y_{ijt} = p_{ij}$ for some time instant $t$). In the schedule $S$, for each machine $i$, we imagine that a job $j$ in $J(i,y^*)$ becomes available for $S$ at the time $t$ where $y_{ijt}=p_{ij}$. We schedule the jobs in $S$ (after they become available) using Shortest Job First (SJF) (where jobs in the same class are viewed as having the same size); for two jobs belonging to same class we schedule the jobs in the order given by $y^*$  \footnote{We can also schedule the jobs in the set $J(i,y^*)$ using SRPT as it is an optimal algorithm for the single machine setting; however, to compare the costs it is more convenient to schedule the jobs using the classes.}.
Let $J_k(i,S)$ denote the set of jobs of class $k$ which are assigned to machine $i$ in schedule $S$, and
let $J(i,S) = \cup_k J_k(i,S)$ denote the set of jobs scheduled by $S$ on $i$. Clearly, $J_k(i,S) = J_k(i,y^*)$.
We also observe that, since jobs within a class are considered in order, for each class $k$ and on each machine $i$, there is at most one job belonging to class $k$ which is partially processed (due to preemptions by jobs of a smaller class). This directly implies the following relation between the fractional and integral flow-time of jobs in $S$.
Let $F^S_j$ denote the flow-time of job $j$ in schedule $S$ and $f^S_j$ denote the fractional flow-time.

\begin{lemma}
\label{lem:frac-inter}
Fix a machine $i$ and the set of jobs belonging to class $k$. Then,

$$
\sum_{j \in J_k(i,S)} F^S_j \leq  \sum_{j \in J_k(i,S)} f^S_j + \sum_{j \in J(i,S)} p_{ij}.
$$
\end{lemma}

\vspace{2mm}
\noindent\textbf{Remark:} Note that first two summations are over $ J_k(i,S)$, while the third summation is over  $J(i,S)$.

\vspace{2mm}

\begin{proof}
We use the alternate view of integral and fractional flow-times. Let $C^S_j$ denote the completion time of job $j$ in the schedule $S$. Then, the integral flow-time of $j$ is $F^S_j = \int^{C^S_j}_{t = r_j} 1\cdot dt$ and the fractional flow-time is $f^S_j = \int^{C^S_j}_{t = r_j} p_{ij}(t)/p_{ij} \cdot dt$, where $p_{ij}(t)$ denotes the remaining processing time of job $j$ on machine $i$.

Let $J_k(i,S,t)$ denote the set of jobs available for processing at time $t$ of class $k$ on machine $i$ in $S$, which have not been completed, and $\mathcal{T}(i,k)$ denote the set of time instants where $J_k(i,S,t) \geq 1$, i.e.~at least one job of class $k$ is alive. Then,

\begin{eqnarray}
\sum_{j \in J_k(i,S)} F^S_j &=& \int_{t \in \mathcal{T}(i,k)} |J_k(i,S,t)| dt \nonumber \\
&\leq& \int_{t \in \mathcal{T}(i,k)} \left(1 + \sum_{j \in  J_k(i,S)} \frac{p_{ij}(t)}{p_{ij}}\right) dt \nonumber \\
&\leq& \sum_{j \in J(i,S)} p_{ij} + \sum_{j \in J_k(i,S)} f^S_j  \nonumber
\end{eqnarray}

The first inequality follows as there is at most one partially processed job of class $k$ at any time in $S$. The second inequality follows by observing that $\int_{t \in \mathcal{T}(i,k)} 1 dt$ is simply the time units when at least one class $k$ job is alive. This can be at most the time when any job (of any class) is alive, which is precisely equal to $\sum_{j \in J(i,S)} p_{ij}$, the total processing done on machine $i$ (as the schedule $S$ is never idle if there is work to be done). Thus, $\int_{t \in \mathcal{T}(i,k)} 1 dt \leq \sum_{j \in J(i,S)} p_{ij}$. Moreover,  $\int_{t \in \mathcal{T}(i,k)}\sum_{j \in  J_k(i,S)} \frac{p_{ij}(t)}{p_{ij}} dt = \sum_{j \in  J_k(i,S)} \int_{t \geq r_j} \frac{p_{ij}(t)}{p_{ij}} dt$ which is exactly the total fractional flow-time $\sum_{j \in J_k(i,S)} f^S_j$.

Let $V_k(y^*, i, t)$ denote the total remaining processing time (or volume) of jobs of class $k$ alive at time $t$ on machine $i$ in the schedule defined by $y^*$ (i.e.~these are precisely the jobs that are released but not yet scheduled by $t$); similarly, let $V_k(S,i,t)$ denote the total remaining processing time of jobs of class $k$ that have $r_j \leq t$, but are unfinished at time $t$ on machine $i$ in the schedule $S$. As a job is available for $S$ only after it is scheduled in $y^*$, we make the following simple observation.

\begin{observation}
\label{obs12}
For any $k$,  $V_k(y^*, i, t) \leq V_k(S, i, t)$. Moreover, $V_k(S,i,t) - V_k(y^*, i, t)$ is the volume of precisely those jobs of class $k$ that are available to $S$ (i.e.~already scheduled in $y^*$), but have not been completed by $S$.
\end{observation}

Using the above observation we show that  $V_k(y^*, i, t)$ and $V_k(S, i, t)$ do not deviate by too much, which is very crucial for our analysis.

\begin{lemma}
\label{lem:backlog}
For every machine $i$, every class $k$, and $\forall t, V_k(S, i, t) - V_k(y^*,i,t) \leq  O(\log n)\cdot 2^k $
\end{lemma}

\begin{proof}
By Observation \ref{obs12}, $V_k(S, i, t) - V_k(y^*,i,t)$ is the total processing time of jobs of class $k$ that are available for processing in $S$ at time $t$ and not yet completed. As $ V_k(S, i, t) - V_k(y^*,i,t) \leq V_{\leq k}(S, i, t) - V_{\leq k}(y^*,i,t)$ (this follows by Observation \ref{obs12} as $V_{k'}(S,i,t) \geq V_{k'} (y^*,i,t)$ for each $k'$), it suffices to bound the latter difference.
Let $t' \leq t$ be the last time before $t$ when machine $i$ was idle in $S$, or was processing a job of class strictly greater than $k$. This means that no jobs of class $\leq k$ are available to $S$ (as they have either not arrived or have not yet been made available by $y^*$). Thus, $V_{\leq k}(S,i,t') = V_{\leq k}(y^*,i,t')$  or equivalently $V_{\leq k}(S,i,t') - V_{\leq k}(y^*,i,t')=0$.
By the low overload property, the total volume of jobs belonging to class at most $k$ that becomes available during $(t',t]$ is at most $(t-t') + O(\log n) 2^k$. Since $S$ processes only jobs of class at most $k$ during $(t',t]$ (by definition of $t'$), $S$ completes precisely $(t-t')$ volume of jobs belonging to class at most $k$.  This implies $V_{\leq k}(S,i,t) - V_{\leq k}(y^*,i,t)= O(\log n) 2^k$.
\end{proof}

We are now ready to show how this implies Theorem \ref{thm:sumofflowtimes}

\begin{proof} (Theorem \ref{thm:sumofflowtimes})
We first compare the fractional flow-times of schedules defined by $y^*$ and $S$ and then use Lemma \ref{lem:frac-inter} to complete the argument.

Define $y^S_{ijt}$ variables  corresponding to the schedule $S$ by setting $y^S_{ijt}$ to the amount of processing done on job $j$ on machine $i$ at time $t$ in the schedule $S$ .
Let $P(S,i) = \sum_{j \in J(i,S)} \sum_t  y^S_{ijt}$ denote the total processing time of the jobs scheduled on machine $i$ in $S$. Clearly, since the set of jobs on machine $i$ in $y^*$ and $S$ is identical, we have $P(S,i) = P(y^*, i)$. Let
 $\mathcal{T}(i,k)$ be the times when there is at least one available but unfinished job in $S$. Recall that
$\int_{t \in \mathcal{T}(i,k)} 1 \cdot dt = P(i,S)$.

Then, the difference between the fractional flow-times of jobs in $S$ and $y^*$ can be bounded by

\begin{eqnarray}
\sum_j (f_j^S - f_j^{y^*}) & = &    \sum_{i} \sum_{t} \sum_{k} \sum_{j: p_{ij} \in (2^{k-1}, 2^k]} (y^S_{ijt} - y^*_{ijt}) \cdot  \left(\frac{t-r_j }{p_{ij}}  \right)  \nonumber \\  
 &\leq&  \sum_{i} \sum_{t} \sum_{k} \sum_{j: p_{ij} \in (2^{k-1}, 2^k]} (y^S_{ijt} - y^*_{ijt}) \cdot  \left(\frac{t-r_j }{2^{k-1}}  \right) \nonumber \\
 &=& \sum_i \sum_t \sum_{k} \sum_{j: p_{ij} \in (2^{k-1}, 2^k]} \frac{1}{2^{k-1}} (V_{k}(S,i,t)- V_{k}(y^*,i,t)) \label{flowconn} \\
 &\leq& \sum_i \sum_k \sum_{t \in \mathcal{T}(i,k)} O(\log n) = \sum_i \sum_k O(\log n) P(i,S) \qquad  [\textrm{By Lemma (\ref{lem:backlog}})]  \nonumber \\
 &\leq& \sum_i O(\log n \cdot \log P) P(i,S)  =  O(\log n \cdot \log P) P(S)\nonumber
 \end{eqnarray}

Here, the equation \eqref{flowconn} follows as for any schedule $S$,

$$\sum_{j: p_{ij} \in (2^{k-1}, 2^k]} \sum_{t \geq r_j } y^S_{ijt} \cdot (t-r_j) = \sum_{t} V_{k}(S,i,t)$$

by the two different ways of looking at fractional flow-time. Next, we can bound the total flow-time as

 \begin{eqnarray*}
 \sum_{j} F_j^S  & = &  \sum_{i} \sum_{k} \sum_{j \in J_k(i,S)} F_j^S  \\
&\leq&  \sum_{i} \sum_{k} \left(\sum_{j\in J_k(i,S)} f_j^S + \sum_{j \in J(i,S) } p_{ij}\right)  \qquad [\textrm{By Lemma (\ref{lem:frac-inter}})]\\
  & = &  \sum_j f_j^S  +  \sum_{i} \sum_{k} \sum_{j \in J(i,S) } p_{ij} \\
 & \leq &  \sum_j f_j^S  +  O(\log P) P(S) \\
&\leq&  \sum_j f_j^{y^*} + O(\log n \cdot  \log P) P(S)
  \end{eqnarray*}

which is at most  $O(\log n \cdot  \log P)$ times the value of optimal solution to \ref{Newprimal}.
\end{proof}

\subsection{Iterated Rounding of \lpp and Proof of Lemma \ref{lem:itr}}
\label{s:iteratedrounding}

In this section we prove the Lemma \ref{lem:itr} using iterated rounding. In the iterated rounding technique, we successively relax the \lpp with a sequence of linear programs, each having fewer constraints than the previous one, while ensuring that optimal solutions to the linear programs have costs that is at most the cost of an optimal solution to \lpp. An excellent reference for  various applications of the technique is \cite{Lau}.

We denote the successive relaxations of \lpp  by $LP(\el)$ for $\el = 0,1,\ldots$.
Let $J(\ell)$ denote the set of jobs that appear in $LP(\el)$.
The linear program $LP(0)$ is same as \lpp, and $J(0) = J$.
We define $LP(\el)$ for $\el > 0$ inductively as follows.

\begin{itemize}
\item {\em Computing a basic optimal solution:} Find a basic optimal solution $y^*(\el-1) = \{y^{\el-1}_{ijt}\}_{i,j,t}$ to $LP(\el-1)$. We use $y^{\el-1}_{ijt}$ to indicate the value taken by the variable $y_{ijt}$ in the solution $y^*(\el-1)$. Let $\mathcal{S}_{\el-1}$ be the set of variables in the support of $y^*(\el-1)$. We initialize $J(\ell) = J(\ell-1)$.

\item {\em Eliminating 0-variables:} The variables $y_{ijt}$ for $LP(\el)$  are defined only for the variables in $\mathcal{S}_{\el-1}$. That is, if $y^{\el-1}_{ijt} = 0$ in $y^*(\el-1)$, then these variables are fixed to 0 forever, and do not appear in $LP(\el)$.

\item {\em Fixing integral assignments:} If a variable $y^{\el-1}_{ijt} = p_{ij}$ in $y^*(\el-1)$ for some job $j$, then $j$ is permanently assigned to machine $i$ at time $t$ in $y^*$ (as required by Lemma \ref{lem:itr}), and we update $J(\ell)=J(\ell)\setminus\{j\}$. We drop all the variables corresponding to the job $j$ in $LP(\el)$, and also drop the service constraint (\ref{eqn:serviceL}) for the job $j$. We use $A(\el-1)$ to denote the set of jobs which get integrally assigned in $(\el-1)$-th iteration. We redefine the intervals based on the unassigned jobs next.

\noindent \textbf{Remark:} It will be convenient below not to view an interval as being defined by its start and end times, but by the $y_{ijt}$-variables it contains.

\item {\em Defining intervals for $\el$-th iteration:}
Fix a class $k$ and machine $i$. We define the new intervals $I(i,*,k,\el)$ and their sizes as follows.

Consider the jobs in $J(\el)$ (those not yet integrally assigned) belonging to classes $\leq k$, and order the variables $y_{ijt}$ in increasing order of $t$ (in case of ties, order them lexicographically).
Greedily group consecutive $y_{ijt}$ variables (starting from the beginning) such that sum of the $y^{\el-1}_{ijt}$ values of the variables in that group first exceeds $4 \cdot 2^k$.

Each such group will be an interval (which we view as a subset of $y_{ijt}$ variables). Define the size of an interval $I=I(i,*,k,\el)$ as
\begin{equation}
\label{sizedef}
 \size(I) = \sum_{y_{ijt} \in I} y^{\el-1}_{ijt}.
 \end{equation}
As $y^{\el-1}_{ijt} \leq 2^k$ for jobs of class $k$, clearly $\size(I) \in [4\cdot 2^k, 5 \cdot 2^k]$ for each $I$ (except possibly the last, in which case we can add a couple of extra dummy jobs at the end) .

\end{itemize}

Note that the intervals formed in $LP(\ell)$ for $\ell>0$ are not (exactly) related to time anymore (unlike $LP(0)$), and in particular, can span much longer duration of time than $4 \cdot 2^k$. All we ensure is that the amount of unassigned volume in an interval is $\Omega(2^k)$.

\vspace{2mm}

{\bf Defining the LP for $\el$-th iteration:}

With the above definition intervals $I(i,a,k,\el)$ and the $y_{ijt}$ variables defined for the $\el$-th iteration, we write the linear programming relaxation for $\el$-th round, $LP(\el)$.

\begin{align}
	\text{Minimize} \sum_{i} \sum_{t \geq r_{j}} \sum_{k} \sum_{j \in J(\el): j \in (2^{k-1}, 2^k]} \left(\frac{t- r_j}{p_{ij}} + \frac{1}{2} \right) \cdot y_{ijt}   \tag{\lpl} \label{LP(l)}
\end{align}
\vspace{-1mm}
\begin{align}
	&& \textrm{ s.t. }\quad \sum_{i} \sum_{t \geq r_j}\frac{ y_{ijt}}{p_{ij}} &\geq 1&& \forall j \in J(\el) \label{eqn:serviceL} \\
	&& \sum_{\substack{y_{ijt} \in I(i,a,k,\el)}} y_{ijt} &\leq \size(I(i,a,k,\el)) &&\forall i, k, a  \label{eqn:capacityL} \\
	&& y_{ijt} &\geq 0  &&\forall i,  j \in J(\el),  t \, : \, t \geq r_{j} \nonumber
\end{align}

\subsubsection{Analysis}

We note that $LP(\ell)$ is clearly a relaxation of $LP(\ell-1)$ (restricted to variables corresponding to jobs in $J(\el)$). This follows as setting $y_{ijt} = y^{\el-1}_{ijt}$ is a feasible solution for $LP(\el)$ (by the definition of $\size(I)$). Moreover, the objective function of $LP(\ell)$ is exactly the objective of $LP(\ell-1)$ when restricted to the variables in $J(\el)$.
Let $y^*$ denote the final integral assignment (assuming it exists) obtained by applying the algorithm iteratively to $LP(0),LP(1),\ldots$. Then this implies

\begin{lemma}
\label{lem:cost}
The cost of the integral assignment $cost(y^*)$ is at most the cost of optimal solution to \lpp.
\end{lemma}

\vspace{2mm} 
{\bf Bounding the number of iterations:}

We now show that the sequence of $LP(\el)$ relaxations terminate after some small number of rounds.
Let $N_{\el} = |J(\el)|$ denote the number of jobs in $LP(\el)$ (i.e.~the one unassigned after solving $LP(\el-1)$).

\begin{lemma}
\label{lem:dec}
After each iteration, the number of unassigned jobs decreases by a constant factor. In particular, for each $\ell$: $N_{\el} \leq  N_{\el-1}/2$.
\end{lemma}

\begin{proof}
Consider a basic optimal solution $y^*(\el-1)$ to $LP(\el-1)$. Let $\mathcal{S}_{\el-1}$ denote the non-zero variables in this solution, i.e.~$y^{\el-1}_{ijt}$ such that $y^{\el-1}_{ijt} > 0$. Consider a linearly independent family of tight constraints in $LP(\ell-1)$ that generate the solution $y^*(\el-1)$. As tight constraints $y^{\el-1}_{ijt}=0$ only lead to $0$ variables, it follows that
$|\mathcal{S}_{\el-1}|$ is at most the number of tight constraints  \eqref{eqn:serviceL} or tight capacity constraints \eqref{eqn:capacityL}. Let $C_{\ell-1}$ denote the number of tight capacity constraints. Thus,
\begin{equation}
\label{eq:sub23}
 |\mathcal{S}_{\el-1}| \leq N_{\ell-1} + C_{\ell-1}.
 \end{equation}
Recall  that $A(\ell-1)$ denotes the set of jobs that are assigned integrally in the solution $y^*(\el-1)$.
As each job not in $A(\ell-1)$ contributes at least two to $|\mathcal{S}_{\el-1}|$, we also have
\begin{equation}
 |\mathcal{S}_{\ell-1}| \geq |A(\ell-1)| + 2 (N_{\ell-1} - |A(\ell-1)|)  = N_{\ell-1} + N_{\ell}.
\end{equation}
The equality above follows as $N_{\ell} = N_{\ell-1} - |A(\ell-1)|$ is the number of the (remaining) jobs considered in $LP(\ell)$.
Together with \eqref{eq:sub23} this gives
\begin{equation}
\label{eq:rel123}
N_{\ell} \leq C_{\ell-1}.
\end{equation}
We now show that $C_{\ell-1} \leq N_{\ell-1}/2$, which together with \eqref{eq:rel123} would imply the claimed result. We do this by a charging scheme. Assign two tokens to each job $j$ in $N_{\ell-1}$. The jobs redistribute their tokens as follows.

Fix a job $j$ and let $k(i)$ denote the class of $j$ on machine $i$.
For each machine $i$, time $t$ and class $k' \geq k(i)$,
the job $j$ gives $\frac{1}{2^{k'-k(i)}} \frac{y^{\el-1}_{ijt}}{ p_{ij}}$ tokens to the class $k'$ interval $I(i,a,k',\el-1)$ on machine $i$ containing $y_{ijt}$.
If there are multiple time slots $t$ in an interval $I(i,a,k',\el-1)$ with $y^{\el-1}_{ijt}>0$, then $I(i,a,k',\el-1)$ receives a contribution from each of these slots. This is a valid token distribution scheme as the total tokens distributed by the job $j$ is at most

\begin{eqnarray*}
\sum_{i} \sum_{t} \sum_{k'\geq k(i)}  \frac{y^{\el-1}_{ijt}}{2^{k'-k(i)} \cdot p_{ij}} &=& \sum_i \sum_t  \left( \frac{y^{\el-1}_{ijt}} {p_{ij}} \cdot \sum_{k'\geq k(i)} \frac{1}{2^{k'-k(i)}} \right) \\
&\leq& 2 \cdot \sum_i \sum_t \frac{y^{\el-1}_{ijt}} {p_{ij}}  = 2.
\end{eqnarray*}

Next, we show that each tight constraint of type \eqref{eqn:capacityL} receives at least $4$ tokens.
If an interval $I(i,a,k',\el-1)$ of class $k'$ on machine $i$ is tight, this means that
$$ \sum_{y_{ijt} \in I(i,a,k',\el-1)} y^{\el-1}_{ijt} = \size(I(i,a,k',\el-1))$$ which is at least $  4 \cdot 2^{k'}.$
Now, the tokens given by a variable $y_{ijt}$ in $ I(i,a,k',\el-1)$ where $j$ is of class $k(i) \leq k'$ are
$$ \frac {y^{\el-1}_{ijt}}{(2^{k'-k(i)} \cdot p_{ij})} \geq \frac{y^{\el-1}_{ijt}}{(2^{k'-k(i)} \cdot  2^{k(i)})} = \frac{y^{\el-1}_{ijt}}{2^{k'}}.$$

Thus, the tokens obtained by $I(i,a,k',\el-1)$ are at least $$\sum_{y_{ijt} \in I(i,a,k',\el-1) } y^{\el-1}_{ijt}/2^{k'} \geq 4 \cdot 2^{k'}/2^{k'} = 4.$$
 As each job distributes at most 2 tokens and each tight interval receives at least 4 tokens, we conclude that $C_{\ell-1} \leq N_{\ell-1}/2$.
\end{proof}

\vspace{3mm}

{\bf Bounding the backlog:} To complete the proof of Lemma \ref{lem:itr}, it remains to show that for any time period $[t_1,t_2]$ and for any class $k$, the total volume of jobs belonging to class at most $k$ assigned to $[t_1,t_2]$ in $y^*$ is at most $(t_2 - t_1) + O(\log n) 2^k$. Recall that $A(\el)$ denotes the set of jobs which get integrally assigned in the $\el$-th round. We use $A(t_1,t_2,i,k,\el)$ to denote the set of jobs of class $\leq k$ which get integrally assigned to the machine $i$ in the interval $[t_1, t_2]$.

Given the solution $y^*(\el)$ to \lpl and a time interval $[t_1,t_2]$, let us define
\begin{eqnarray*}
\vol(t_1,t_2,i,k,\el):= \sum_{j \in J(\el): p_{ij} \leq 2^k} \sum_{t \in [t_1,t_2]} y^{\el}_{ijt}  \\
+ \sum_{\ell' \leq (\ell-1)} \sum_{j \in A(t_1,t_2,i,k,\ell')} p_{ij} 
\end{eqnarray*}
as the total size of jobs of class $\leq k$, assigned either integrally or fractionally to the period $[t_1,t_2]$ after $\ell$ rounds. The following key lemma controls how much $\vol$ can get worse in each round.

\begin{lemma}
\label{lem:freespace}
For any period $[t_1,t_2]$, machine $i$, class $k$, and round $\el$,
$$\vol(t_1,t_2,i,k,\el)   \leq O(1) \cdot 2^k + \vol(t_1,t_2,i,k,\el-1).$$
\end{lemma}

\begin{proof}
By the definition of $\vol$ this is equivalent to showing that
  \begin{eqnarray}
   \sum_{j \in J(\el): p_{ij} \leq 2^k} \sum_{t \in [t_1,t_2]} y^{\el}_{ijt} + \sum_{j \in A(t_1,t_2, i,k,\ell-1)} p_{ij} \nonumber \\
\qquad \leq O(1) \cdot 2^k  + \sum_{j \in J(\el-1): p_{ij} \leq 2^k} \sum_{t \in [t_1,t_2]} y^{\el-1}_{ijt}  \label{eqn:main342}
  \end{eqnarray}

Fix a time period $[t_1,t_2]$. The main idea is that in each round $\el$, the error to $\vol$ can be introduced only due to the two class $k$ intervals overlapping with the boundary of $[t_1,t_2]$.

Consider the maximal set of contiguous intervals $I(i,b,k,\el)$, 
$I(i,b+1,k,\el),\ldots I(i,b+h,k,\el)$, for some $b,h\geq 0$, that contain the period $[t_1,t_2]$. More precisely, $b$ is the smallest index such that $I(i,b,k,\el)$ contains some $y_{ijt}$  with $t\in[t_1,t_2]$, and $h$ is the largest index such that $I(i,b+h,k,\el)$ contains some $y_{ijt}$ with $t\in[t_1,t_2]$.
 As these intervals have size at most $5  \cdot 2^k$, we have
\begin{equation}
\label{eqn:free1}
\displaystyle \sum_{y_{ijt} \in I(i,b,k,\el)} y^{\el}_{ijt} + \sum_{y_{ijt} \in I(i,b+h,k,\el)} y^{\el}_{ijt} \leq 10 \cdot 2^k.
\end{equation}

Now, consider the intervals $I(i,b',k,\el) \in \{ I(i,b+1,k,\el), $ $I(i,b+2,k,\el),\ldots I(i,b+h-1,k,\el)\}$ that are completely contained in $[t_1,t_2]$ (i.e.~for all  $y_{ijt} \in I(i,b',k,\el)$, $t \in [t_1, t_2]$). By definition of these intervals and capacity constraints of \lpl we have,

\begin{eqnarray}
\sum^{b+h-1}_{b'= b+1} \sum_{y_{ijt} \in I(i,b',k,\el)} y^{\el}_{ijt} &\leq& \sum^{b+h-1}_{b'= b+1} \size(I(i,b',k,\el))  \nonumber \\
&\leq&  \sum^{b+h-1}_{b'= b+1}  \sum_{y_{ijt}\in  I(i,b',k,\el)}  y^{\el-1}_{ijt} \nonumber \\
&\leq&  \sum_{j \in J(\el): p_{ij} \leq 2^k} \sum_{t \in [t_1,t_2]} y^{\el-1}_{ijt} \qquad \label{eqn:free4}
\end{eqnarray}

The first inequality follows from the constraints \eqref{eqn:capacityL} of \lpl, where as the second one follows from the definition \eqref{sizedef} of \size. We now prove \eqref{eqn:main342}. Consider,

\begin{eqnarray}
\displaystyle
\sum_{j \in J(\ell): p_{ij} \leq 2^k} \sum_{t \in [t_1,t_2] }y^{\ell}_{ijt} &\leq&    \sum^{b+h}_{b'= b} \sum_{y_{ijt} \in I(i,b',k,\el)} y^{\el}_{ijt}  \nonumber \\
&\leq&  10 \cdot 2^k  + \sum_{j \in J(\el): p_{ij} \leq 2^k} \sum_{t \in [t_1,t_2]} y^{\el-1}_{ijt} \quad [\textrm{by \eqref{eqn:free1} and \eqref{eqn:free4}}]  \nonumber \\
 &\leq&   10 \cdot 2^k + \sum_{j \in J(\el-1): p_{ij} \leq 2^k} \sum_{t \in [t_1,t_2]} y^{\el-1}_{ijt} -  \sum_{j \in A(t_1,t_2,i,k,\ell-1)} p_{ij}  \nonumber
\end{eqnarray}

The last step follows as $J(\el) = J(\el-1) \setminus A(\el-1)$ and as $$\sum_{j \in A(t_1,t_2,i,k,\ell-1)} y^{\el-1}_{ijt} = \sum_{j \in A(t_1,t_2,i,k,\ell-1)} p_{ij}.$$
\end{proof}

This directly implies the following bound on the total error in any period $[t_1, t_2]$ in $y^*$.

\begin{lemma}
\label{lem:error}
For a given time period $[t_1, t_2]$, machine $i$ and class $k$, the total volume of jobs of class at most $k$, assigned to the interval is at most $(t_2 - t_1) + O(\log n) 2^k$.
\end{lemma}

\begin{proof}
Recall the definition of an interval $I(i,a,k,0)$ in $LP(0)$. Each interval $I(i,a,k,0)=(t',t'']$ has size $4 \cdot 2^k$ and contains all the $y_{ijt}$ variables for jobs of class at most $k$ and $t\in (t',t'']$. Therefore, for any period $[t_1,t_2]$, by considering the capacity constraints ($\ref{eqn:capacity}$) of $LP(0)$ for the overlapping intervals $I(i,*,k,0)$, we obtain
\begin{eqnarray}
\label{eqn:recur2}
\displaystyle \vol(t_1,t_2,i,k,0) & = &  \sum_{j: p_{ij} \leq 2^k} \sum_{t \in [t_1,t_2]} y^{0}_{ijt} \nonumber \\
& \leq &  (t_2 - t_1) + O(1)\cdot 2^k \nonumber
\end{eqnarray}

Applying lemma \ref{lem:freespace} inductively (for the term $\vol$ in the above equation) over the $O(\log n)$ iterations of the algorithm gives the result.
\end{proof}

\begin{proof}(Lemma \ref{lem:itr})
Consider the final solution $y^*$ at the end of the algorithm. By our construction each job is integrally assigned in $y^*$. By Lemma (\ref{lem:cost}), cost($y^*$) is no more than the cost of an optimal solution to \lpp. By Lemma (\ref{lem:error}), for any time period $[t_1, t_2]$, machine $i$ and class $k$, the total volume of jobs assigned of jobs in class $\leq k$ is at most $(t_2 - t_1) + O(\log n) 2^k$. This concludes the proof.
\end{proof}

\endgroup

\subsection{The  $O(\log^2 n)$ approximation}
\label{pton}
The $O(\log^2 n)$ approximation follows directly by observing that jobs much small $p_{\max}$ essentially have no effect. 
 
The algorithm guesses $p_{\max}$, the value of the maximum job size in an optimal solution (say, by trying out all possible $mn$ choices), and considers a modified instance $J'$ where we set $p_{ij}=p_{\max}/n^2$ whenever $p_{ij}<p_{\max}/n^2$, and applies the previous algorithm for $J'$.
Clearly, $P \leq n^2$ for $J'$. Moreover $\textrm{OPT}(J') \leq 2\ \textrm{OPT}(J)$.
Indeed, consider the optimum solution for $J$ and for each job $j$ assigned to machine $i$ with size $p_{ij}< p_{\max}/n^2$, increase its size  to $p_{\max}/n^2$ and push all the jobs behind it by the amount by which the size increases. This gives a valid schedule for $J'$.
Each job can be pushed by at most $n$ jobs, and hence its flow time increases by at most $n \cdot p_{\max}/n^2$. Thus the total flow-time increases by at most $p_{\max}$ which is at most $\textrm{OPT}(J)$.

\section{Minimizing the Maximum Flow-time}
\label{sec:max-flow}	

Now, we consider the problem of minimizing the maximum flow-time.  By doing a binary search, we assume that we know the value of an optimum solution (OPT); say OPT = $D$. Let us index the jobs by their release times (breaking ties arbitrarily).

We write a linear programming relaxation for the problem. In this relaxation, there is a variable $x_{ij}$ denoting the total processing done on a job $j$ on a machine $i$. If $p_{ij} > D$ for a job $j$ on a machine $i$, then we set $x_{ij}=0$, as $j$ cannot be scheduled on $i$. The first set of constraints (\ref{eqn:service1}) ensure that each job is completely processed.
To see the second constraint (\ref{eqn:capacity2}),
we note that any job released during the interval $[t,t']$ must be completed by time $t'+D$. Thus, the total
size of the jobs released in $[t,t']$ that are assigned to $i$ can be at most $(t'-t)+D$. Moreover, it suffices to consider intervals such that $t,t'$ are release dates of some jobs (as this gives the tightest constraints).

\begin{align}
	 &&  \sum_{i} \frac{x_{ij}} {p_{ij}} &\geq 1&& \forall j \label{eqn:service1}\\
	&&   \sum_{r_j \in [t, t']} x_{ij}  &\leq (t' - t) + D  &&\forall i, \forall  t, t ' \in \{r_1, \ldots, r_n\} \label{eqn:capacity2} \\
	&& x_{ij} &\geq 0 &&\forall i,j   \label{eqn:ng3} \\
&& x_{ij} &= 0 &&\forall i,j \quad \textrm { with } p_{ij} > D.
\end{align}	

\vspace{2mm}
\noindent\textbf{Remark:} Note that the variables $x_{ij}$ do not specify the time at which the job $j$ is assigned to the machine $i$. However, it is instructive to view $x_{ij}$ units of work being assigned at the time $r_j$ (the release time of $j$).

\vspace{2mm}

We say that a job is {\em integrally} assigned to machine $i$ in the interval $[t_1, t_2]$ if $x_{ij} = p_{ij}$ and $r_j \in [t_1, t_2]$. Similarly,  if $x_{ij} > 0$ and $x_{ij} \neq p_{ij}$, then the job is assigned fractionally to the machine $i$. Let $p_{\max}$ denote the maximum value of $p_{ij}$ in some optimum schedule (note that $p_{ij} \leq D$).
For convenience, let us assume that the release times are distinct (say, by perturbing them by some infinitesimally small amount).

As previously, we prove Theorem \ref{thm:maxflow} using iterated rounding.
To this end, we will show how to create a ``tentative" schedule satisfying the following properties.

\begin{lemma}
\label{lem:errormax}
There exists a solution $x^* = \{x_{ij}\}_{i,j}$ with the following properties:
\begin{itemize}
\item $x^*$ {\em integrally assigns} each job $j$ to a single machine $i$; i.e., $x_{ij} $ is equal to $p_{ij}$ for some machine $i$.
\item For any time interval $[t_1,t_2]$, the total volume of the jobs assigned in $x^*$ is at most $(t_2 - t_1)+D + O(\log n) \cdot p_{\max} $. That is,
$$ \sum_{j: r_j \in [t_1, t_2]} x_{ij} \leq (t_2 - t_1) + D+ O(\log n) \cdot p_{\max}.$$
\end{itemize}
\end{lemma}	

We first show that Theorem \ref{thm:maxflow} follows easily from the above lemma.

\begin{proof}(Theorem \ref{thm:maxflow})
Given a solution $x^*$ satisfying the properties of Lemma \ref{lem:errormax}, we construct a valid schedule such that the flow-time of each job is at most $D+O(\log n) \cdot p_{\max}$.
Fix a machine $i$. Consider the jobs $J(i, x^*) = \{j \hspace{1mm} | \hspace{1mm} x_{ij} = p_{ij}\}$ assigned to machine $i$, and schedule them in First In First Out (FIFO) order.

Fix a job $j$. Consider the interval $[0,r_j]$, and let $t' \in [0, r_j]$ be the latest time instant when the machine $i$ is idle.  This implies that all the jobs in $J(i,x^*)$ released in the interval $[0,t']$ are completed by $t'$. As the machine is busy during $(t', r_j]$ and the total volume of jobs assigned in the interval is at most $(r_j - t' ) + D+ O(\log n)  \cdot p_{\max}$ (as promised by Lemma \ref{lem:errormax}),  the total volume of the jobs alive at $r_j$ is at most $D + O(\log n)  \cdot p_{\max}$ . As we schedule the jobs using FIFO, the job completes by time $r_j + D + O(\log n)  \cdot p_{\max}$.
 \end{proof}

Henceforth, we focus on proving Lemma \ref{lem:errormax}.

\subsection {Iterated Rounding and Proof of Lemma \ref{lem:errormax}}

We prove Lemma \ref{lem:errormax} using iterated rounding. Similar to the proof of Lemma (\ref{thm:sumofflowtimes}), we write successive relaxations of the LP (\ref{eqn:service1}-\ref{eqn:ng3}) denoted by $LP(\el)$ (\ref{eqn:1service}-\ref{eqn:3ng}), for $\el = 0,1,2...$, such that number of constraints drop by a constant fraction in each iteration. Finally, we obtain a solution where each job is integrally assigned to a single machine. $LP(0)$ is same as LP (\ref{eqn:service1}-\ref{eqn:ng3}). Let $J(\el)$ denote the set of jobs which are yet to be integrally assigned at the beginning of iteration $\el$. Let $J(0) = J$. Next, we define $LP(\el)$ for $\el \geq 1$.

\begin{itemize}

\item {\em Computing a basic feasible solution:} Solve $LP(\el-1)$ and find a basic feasible solution $x^*(\el-1) = \{x^{\el-1}_{ij}\}_{i,j}$ to $LP(\el-1)$. We use $x^{\el-1}_{ij}$ to indicate the value taken by the variable $x_{ij}$ in the solution $x^*(\el-1)$. Initialize $J(\el) = J(\el-1)$.

\item {\em Eliminating zero variables:} Variables $x_{ij}$ of $LP(\el)$  are defined  by the set of positive variables in the basic feasible solution to $LP(\el-1)$. In other words, if $x^{\el-1}_{ij} = 0$ in $x^*(\el-1)$, then $x_{ij}$ is not defined in $LP(\el)$.

\item {\em Fixing integral assignments:} If $x^{\el-1}_{ij} = p_{ij}$ for some job $j$, then the job $j$ is permanently assigned to the machine $i$ in the solution $x^*$, and we update $J(\ell)=J(\ell)\setminus\{j\}$.
We drop all the variables involving the job $j$ in $LP(\el)$, and the constraint (\ref{eqn:1service}). Moreover, we update the constraints of type (\ref{eqn:1capacity}) as described next.

\item {\em Defining Intervals:} For each machine $i$ and for each iteration $\el$, we define the notion of intervals $I(i,a,\el)$ as follows:  Consider the variables $x_{ij}$  for jobs $j \in J(\el)$ (i.e.~the ones not assigned integrally thus far), in the order of non-decreasing release times.
    Greedily group consecutive $x_{ij}$ variables (starting from the beginning) such that sum of the $x^{\el-1}_{ij}$ values in that group first exceeds $2p_{\max}$.
      We call these groups intervals, and denote the $a$-th group by $I(i,a,\el)$. We say $j \in I(i,a,\el)$ if $ x_{ij} \in I(i,a,\el)$, and define $$\size(I(i,a,\ell)) = \sum_{j \in I(i,a,\el)} x_{ij}^{\el-1}.$$

Note that $\size(I(i,a,\ell))\in [2 \cdot p_{\max},3\cdot p_{\max})$ (except possibly for the last interval, in which case we add a dummy job of size $2p_{\max}$.)
\end{itemize}

\vspace{2mm}

{\bf LP($\ell$):} We are now ready to write $LP(\el)$.
\begin{align}
	 &&  \sum_{i} \frac {x_{ij}}{p_{ij}} &\geq 1&& \forall j \in J(\el) \label{eqn:1service}\\
	&&   \sum_{j \in I(i,a,\el)} x_{ij}  &\leq \size(I(i,a,\el)) &&\forall i,a,\el \label{eqn:1capacity} \\
	&& x_{ij} &\geq 0 &&\forall i,j \geq 0  \label{eqn:3ng}
\end{align}	
 By the definition of intervals and their sizes, it is clear that the feasible solution $x^*({\el-1})$ to $LP(\el-1)$ is also a feasible solution to $LP(\el)$. Next, we show that each job is integrally assigned after $O(\log n)$ iterations.

\vspace{3mm}

{\bf Bounding the number of iterations:}
Let $N_{\el}$ denote the number of jobs during the $\el$-th iteration.
\begin{lemma}
\label{lem:conmaxflow}
For  all $\el > 1$, $N_{\el} \leq \frac{N_{\el-1}}{2}$.
\end{lemma}
\begin{proof}
Consider a basic optimal solution $x^*(\el-1)$ to $LP(\el-1)$. Let $\mathcal{S}_{\el-1}$ denote the non-zero variables in this solution, i.e.~$x_{ij}$ such that $x^{\el-1}_{ij} > 0$. Consider a linearly independent family of the tight constraints in $LP(\ell-1)$ that generate the solution $x^*(\el-1)$. Since the tight constraints of type $x^{\el-1}_{ij}=0$ only lead to $0$ variables, it follows that
$|\mathcal{S}_{\el-1}|$ is at most the number of tight service constraints  \eqref{eqn:1service} or tight capacity constraints \eqref{eqn:1capacity}. Let $C_{\ell-1}$ denote the number of tight capacity constraints. Thus,
\begin{equation}
\label{eq:sub}
 |\mathcal{S}_{\el-1}| \leq N_{\ell-1} + C_{\ell-1}
 \end{equation}
Recall  that $A(\ell-1)$ denotes the set of jobs that are assigned integrally in the solution $x^*(\el-1)$. Then, $N_{\ell} = N_{\ell-1} - |A(\ell-1)|$ is the number of remaining jobs that are considered in $LP(\ell)$. As each job not in $A(\ell-1)$ contributes at least a value of two to $|\mathcal{S}_{\el-1}|$, we also have
\begin{equation}
 |\mathcal{S}_{\ell-1}| \geq |A(\ell-1)| + 2 (N_{\ell-1}-|A(\ell-1)|)  = N_{\ell-1} + N_{\ell}
\end{equation}
Together with \eqref{eq:sub} this gives
\begin{equation}
\label{eq:rel}
N_{\ell} \leq C_{\ell-1}
\end{equation}

We now show that $C_{\ell-1} \leq N_{\ell-1}/2$, which together with \eqref{eq:rel} would imply the claimed result.
We know that size of each interval in  $(\el-1)$-th iteration is at least $2 \cdot p_{\max}$. As each tight interval $I(i,a,\el-1)$ has $$\sum_{j \in I(i,a,\el-1)} x^{\el-1}_{ij}  = \size(I(i,a,\el)),$$ we have
$$N_{\ell-1} \geq \frac{\sum_{i,j} x^{\el-1}_{ij}}{p_{\max}} \geq \frac{2 \cdot p_{\max} \cdot C_{\el-1}}{p_{\max}} \geq 2C_{\el-1}$$

Thus we get  $C_{\ell-1} \leq N_{\ell-1}/2$.
\end{proof}

Therefore, the number of jobs which are integrally assigned at each iteration $\el$ is at least $N_{\el}/2$.  Note that number of constraints in $LP(1)$ is at most $n/2$ since size of each interval is at least $2 \cdot p_{\max}$. Hence, the algorithm terminates in $O(\log n)$ rounds.

\vspace{3mm}

{\bf Bounding the overload:}
It remains to show that for any time interval $[t_1,t_2]$, the total size of jobs assigned in the interval $[t_1,t_2]$ in $x^*$ is at most $(t_2 - t_1) + O(\log n)\cdot p_{\max} + D$.

Let $\vol(t_1,t_2,i,\el)$ be the total volume of jobs assigned (both fractionally and integrally) during the period $[t_1,t_2]$ at the end of $\el$-th iteration. Moreover, let $A(t_1,t_2,i,\el-1)$ be the set of jobs assigned in the period $[t_1,t_2]$ in the $(\el-1)$-th iteration, i.e.~$x^{\el-1}_{ij} = p_{ij}$ and $r_{j} \in [t_1,t_2]$.

Given the solution $x^{*}(\el)$ to $LP(\el)$. Clearly,

\begin{equation}
\label{maxfree2}
\vol(t_1,t_2,i,\el) = \sum_{r_j \in [t_1,t_2]} x^{\el}_{ij} + \sum_{\el' < \el} \sum_{ j \in A(t_1,t_2,i,\el') } p_{ij}.
\end{equation}

The following lemma shows that for any time period, the volume does not increase much in each round.

\begin{lemma}
\label{lem:itrmax}
For any iteration $\el$, machine $i$, and any time period $[t_1,t_2]$,
$$\vol(t_1,t_2,i,\el) \leq \vol(t_1,t_2,i,\el-1) +6 \cdot p_{\max}$$
\end{lemma}

\begin{proof}
Consider the maximal contiguous set of intervals $\mathcal{I} = \{ I(i,b,\el), I(i,b+1,\el), \ldots I(i,b+h,\el) \}$ such that for every interval $I(i,b',\el) \in \mathcal{I}$, there exists a job $ j \in I(i,b',\el)$ and $r_j \in [t_1,t_2]$. Recall that size of each interval in $LP(\el)$ is at most $3\cdot p_{\max}$. Hence, the intervals $I(i,b,\el)$ and $I(i,b+h,\el)$ which overlap $[t_1, t_2]$ at the left and right boundaries respectively, contribute at most $6 \cdot p_{\max}$ to the interval $[t_1, t_2]$. Therefore,

\begin{eqnarray}
\sum_{r_j \in [t_1,t_2]} x^{\el}_{ij} & \leq &  \sum^{b+h-1}_{a = b+1}\size(I(i,a,\el)) + 6\cdot p_{\max}\qquad [\textrm{By} \eqref{eqn:1capacity}] \nonumber \\
&\leq& \sum_{r_j \in [t_1, t_2]} x^{\el-1}_{ij} -  \sum_{ j \in A(t_1,t_2,i,\el-1) } p_{ij} + 6 \cdot p_{\max}  \qquad \textrm {[By the interval definition] } \nonumber \\
&\leq& \vol(t_1,t_2,i,\el-1) - \sum_{\el' \leq \el-1} \sum_{ j \in A(t_1,t_2,i,\el') } p_{ij} + 6 \cdot p_{\max}   \qquad [ \textrm{By }\eqref{maxfree2} ]\nonumber
\end{eqnarray}

The lemma now follows by rearranging the terms and using \eqref{maxfree2}.
\end{proof}

\begin{lemma}
\label{lem:volume}
In the solution $x^*$, the total volume of jobs assigned in any interval $[t_1,t_2]$ is at most $(t_2 - t_1) + D+ O(\log n)\cdot p_{\max}$.
\end{lemma}

\begin{proof}
Consider the interval $[t_1,t_2]$. From the constraints of $LP(0)$ over the interval  $[t_1,t_2]$ and the definition of $\vol(i,a,0)$ (equation \ref{maxfree2}), we have,
\begin{eqnarray}
\vol(t_1,t_2,i,0) = \sum_{r_j \in [t_1, t_2]} x^0_{ij} =  &\leq& t_2 - t_1 + D \nonumber
\end{eqnarray}
The result now follows by applying Lemma \ref{lem:itrmax} for the $O(\log n)$ iterations of the algorithm.
\end{proof}

\begin{proof}(Lemma \ref{lem:errormax})
From Lemma \ref{lem:conmaxflow} we know that each job is integrally assigned to a single machine. Lemma \ref{lem:volume} guarantees that the total volume of jobs assigned for every time interval $[t_1, t_2]$ is bounded by $(t_2 - t_1) + D + O(\log n)\cdot p_{\max} $. This gives us the desired $x^*$ and concludes the proof.
\end{proof}

\bibliographystyle{plain}
\bibliography{unrelated}

\end{document}